\documentclass[11pt,twoside]{amsart}
\usepackage{latexsym,amssymb,amsmath}
\newtheorem{theorem}{Theorem}
\newtheorem{corollary}{Corollary}
\newtheorem{remark}{Remark}
\newtheorem{proposition}[theorem]{Proposition}
\newtheorem{definition}[theorem]{Definition} 
\newtheorem{lemma}[theorem]{Lemma}
\newtheorem{example}[theorem]{Example}

\begin{document}

\title[second generalized Hamming weight]{The second generalized
Hamming weight of some evaluation codes arising from a projective torus} 

\author[Gonz\'alez Sarabia]{Manuel Gonz\'alez Sarabia}
\address{Instituto Polit\'ecnico Nacional, 
UPIITA, Av. IPN No. 2580,
Col. La Laguna Ticom\'an,
Gustavo A. Madero C.P. 07340,
 M\'exico, D.F. 
Departamento de Ciencias B\'asicas.
}
\email{mgonzalezsa@ipn.mx}

\author[Camps]{Eduardo Camps}
\address{Instituto Polit\'ecnico Nacional,
ESFM, C.P. 07300, 
M\'exico, D.F.,
Departamento de Ma\-te\-m\'a\-ti\-cas.}
\email{ecfmd@hotmail.com}

\author[Sarmiento]{Eliseo Sarmiento}
\address{Instituto Polit\'ecnico Nacional,
ESFM, C.P. 07300, 
M\'exico, D.F.,
Departamento de Ma\-te\-m\'a\-ti\-cas.}
\email{esarmiento@ipn.mx}

\author[Villarreal]{Rafael H. Villarreal}
\address{Departamento de Matem\'aticas,
Centro de Investigaci\'on y de Estudios Avanzados del IPN,
Apartado Postal 14--740, 07000, 
Ciudad de  M\'exico.}
\email{vila@math.cinvestav.mx}

\begin{abstract}
In this paper we give a  formula for the second generalized
Hamming weight of certain eva\-lua\-tion codes arising from a projective
torus. This allows us to compute the corresponding weights of the
codes parameterized by the edges of a complete
bipartite graph. We determine some of the generalized Hamming weights
of 
non-degenerate evaluation codes arising from
a complete intersection in terms of the minimum distance, the degree
and the $a$-invariant. It is shown that the
generalized Hamming weights and the minimum distance have some 
similar behavior for parameterized codes  These results are used 
to find the complete weight hierarchy of some codes.      
\end{abstract}

\maketitle 

\section{Introduction}
Consider the system of two polynomial equations
\begin{eqnarray} \label{array}
& F_1(X_1,\ldots,X_s)=0, \\
& F_2(X_1,\ldots,X_s)=0, \nonumber
\end{eqnarray}
where $F_1$ and $F_2$ are linearly independent homogeneous
polynomials of degree $d$ in $s$ variables over a finite field
$\mathbb{F}_q$ with $q$ elements.  
The number of solutions of system (\ref{array}) in a projective torus
$\mathbb{T}_{s-1}$ (see Definition \ref{def2}) is given by 
$$
|Z_{\mathbb{T}_{s-1}}(F_1) \cap Z_{\mathbb{T}_{s-1}}(F_2)|,
$$
where $Z_{\mathbb{T}_{s-1}}(F_i)$ is the zero set of $F_i$ for
$i=1,2$ (see Section \ref{finalsection}). In this paper, we find the exact value of the integer
\begin{align*}
\max  \{|Z_{\mathbb{T}_{s-1}}(F_1) \cap Z_{\mathbb{T}_{s-1}}(F_2)|  : \, F_1,  \,&  F_2  \, {\mbox{are linearly independent homogeneous}} \\
& {\mbox{polynomials in}} \, s \, {\mbox{variables}}
\, {\mbox{of degree}} \, d\},
\end{align*}
that is, we solve the following problem: what is the maximum possible
number of solutions of system (\ref{array}) in a projective torus $\mathbb{T}_{s-1}$?

In \cite{Boguslavsky} Boguslavsky answered this question when we
replace the torus by the projective space $\mathbb{P}^{s-1}$: 

\begin{theorem}
The maximum possible number of solutions of system (\ref{array}) in the projective space $\mathbb{P}^{s-1}$ when $d<q-1$ is given by
$$
(d-1)q^{s-2}+p_{s-3}+q^{s-3},
$$
where $p_m=|\mathbb{P}^{m}|=q^m+q^{m-1}+\cdots+q+1$. When $d \geq q+1$ it is known that this number is $p_{s-1}$ \cite{sorensen}.
\end{theorem}

In our case, we are able to obtain (see Theorems \ref{theoremzeroes} and
\ref{otherbound}) that the maximum po\-ssi\-ble number of solutions
of system (\ref{array}) in $\mathbb{T}_{s-1}$ when $d<q-1$ is given by  
$$
(q-1)^{s-2}(d-1)+(q-1)^{s-3},
$$
and when $d \geq q-1$ we can find two appropriate polynomials that vanish at all
points of the projective torus, meaning that this number is
$|\mathbb{T}_{s-1}|=(q-1)^{s-1}$. Moreover, if we also require that $F_1$ and
$F_2$ do not vanish at all points of the torus $\mathbb{T}_{s-1}$, the
maximum number of possible solutions of system (\ref{array}) is given by (see Remarks \ref{dec21-17}
and \ref{finalremark})    
$$
(q-1)^{s-1}-\left\lceil (q-1)^{s-(k+2)}(q-1-l)\right\rceil -\left\lceil (q-1)^{s-(k+3)}(q-2)\right\rceil,
$$
where $s \geq 2$, $d \geq 1$, and $k$, $l$ are the unique integers
such that $d=k(q-2)+l$, $k \geq 0$, $1 \leq l \leq q-2$. These
results allow us to compute the second generalized Hamming weight of
some evaluation codes arising from a projective torus (see Theorem
\ref{theorem3}). This weight should not be confused with the second
Hamming weight, also called next--to--minimal weight, which was
computed by Carvalho in \cite[Theorem 2.4]{carvalho} in a more
general case (affine cartesian codes) when $2 \leq d < q-1$, $s \geq
3$. Recently, in \cite{carvalho2}, Carvalho and Neumann determine
many values of the second least weight codewords for affine cartesian
codes.           

The generalized Hamming weights of a linear code were introduced 
in \cite{helleseth}, \cite{klove}, and rediscovered by Wei in
\cite{wei}. The study of these weights is related to trellis coding,
$t$--resilient functions, and it was motivated from some applications
in cryptography. The weight hierarchy of a code has been examined in
several cases including the following families 
(see \cite{ashikhmin,amoros,GHW2014,guneri,Pellikaan,homma,Johnsen,munuera1,
munuera2,olaya,tsfasman,Geer,Yang} and the references therein):
\begin{enumerate}
\item Golay codes.
\item Product codes.
\item Codes from classical varieties: Reed--Muller codes, codes
from quadrics, Hermitian varieties, Grassmannians, Del Pezzo surfaces. 
\item Algebraic geometric codes.
\item Cyclic and trace codes: BCH, Melas.
\item Codes parameterized by the edges of simple graphs.
\end{enumerate} 

In this work we focus in the so called Reed-Muller-type codes
\cite{duursma,Gold}. 
These codes
are obtained by eva\-lua\-ting the linear space of homogeneous
$d$--forms on a subset of points $X$ of a projective space
$\mathbb{P}^{s-1}$ over a finite field. We denote this linear code by $C_X(d)$ (see
Definition \ref{def1}). If $X$ is the whole projective space we
obtain the projective Reed--Muller codes (see \cite{sorensen}). The
main parameters of $C_X(d)$ were computed in \cite{sarabia} when $X$
is the Segre va\-rie\-ty. The case of a Veronese
va\-rie\-ty was examined in \cite{renteria}. 

The main properties of $C_X(d)$ were studied in detail in 
\cite{duursma} and \cite{Hansen} when the set $X$ is a complete intersection.
In spite of the minimum distance in this case remains unknown, 
in \cite{Ballico}
and \cite{Gold} there are lower bounds for this basic parameter  and in
\cite{soprunov} there is a nice generalization of these results. Other
lower bounds can be found in \cite{min-dis-ci,tohaneanu-vantuyl}. Although
we do not know formulas for the minimum distance for general 
complete intersections, in some particular
cases explicit formulas have been determined (e.g., cartesian codes \cite{hiram},
codes parameterized by a projective torus \cite{ci-codes} 
or by a degenerate projective torus \cite{sarabia5}). One of our main
results gives formulas for some of the generalized Hamming weights
of non-degenerate Reed-Muller-type codes arising from
a complete intersection in terms of the minimum distance, the degree 
and the $a$-invariant (see Theorem \ref{theorem1}).                        

On the other hand, the notion of a code parameterized by a finite set of
monomials was introduced in
\cite{algcodes}. These are Reed-Muller-type codes where $X$ is a subgroup of the
projective torus $\mathbb{T}_{s-1}$. In Section \ref{monomials}, we show that the
generalized Hamming weights and the minimum distance have, in a
certain sense, similar 
behavior for parameterized codes  (see Theorem \ref{theorem2}).
 In particular we recover 
\cite[Proposition 5.2]{algcodes}.

The rest of the contents of this paper are as follows. In Section
\ref{preliminaries}, we introduce the definitions needed to
understand the main results. In Section
\ref{finalsection}, we show another of our main results which gives 
an explicit formula for the second generalized Hamming weight of
the Reed-Muller-type codes arising from a projective torus (see Theorem
\ref{theorem3}). 

If $\mathcal{G}$ is a simple graph (no loops or multiple
edges) and $X$ is the set parameterized by its edges, the code
$C_X(d)$ has been studied in several cases (see \cite{sarabia1},
\cite{sarabia2}, \cite{sarabia3}, \cite{sarabia4}, \cite{sarabia6},
\cite{GHW2014}, \cite{vaz}, \cite{neves}, \cite{algcodes}). As an
application we compute the second generalized
Hamming weight of the codes parameterized by the edges of any
complete bipartite graph (see Remark~\ref{dec-19-17}). Then we use our results  
to find the complete weight hierarchy of the codes
$C_{\mathbb{T}_2}(d)$ over a
finite field with $5$ elements (see Example \ref{example}).                  

\section{Preliminaries} \label{preliminaries}

Let $K=\mathbb{F}_q$ be a finite field with $q$ elements, 
let $\mathbb{P}^{s-1}$ be a projective space over $K$, and let
$X=\{P_1,\ldots,P_{|X|}\}$ be a subset of $\mathbb{P}^{s-1}$ where
$|X|$ is the cardinality of the set $X$. Let $S=K[X_1,\ldots,X_s]=\oplus_{d\geq 0} S_d$ be a polynomial ring with
the standard grading, where $S_d$ is the vector space generated by the
homogeneous polynomials in $S$ of degree $d$. Fix a degree $d\geq 1$. For each $i$ there is
$f_i \in S_d$ such that $f_i(P_i) \neq 0$. Indeed suppose
$P_i=[t_1:\cdots:t_s]$, there is at least one $k \in \{1,\ldots,s\}$
such that $t_k \neq 0$. 
Setting $f_i=X_k^d$ one has that $f_i \in S_d$ and $f_i(P_i) \neq 0$.
Consider the evaluation map       
\begin{align*}
{\rm{ev}}_d: S_d \longrightarrow K^{|X|}, \hspace{1cm} \\
f \mapsto \left( \frac{f(P_1)}{f_1(P_1)},\ldots,\frac{f(P_{|X|})}{f_{|X|}(P_{|X|})}\right).
\end{align*}

This is a linear map between the $K$-vector spaces $S_d$ and $K^{|X|}$.
\begin{definition}\label{def1}
The evaluation code or Reed--Muller-type code of order $d$ 
asso\-ciated to $X$, denoted $C_X(d)$, is the image of ${\rm{ev}}_d$, that is
$$
C_X(d)=\left\{\left(\frac{f(P_1)}{f_1(P_1)},\ldots,
\frac{f(P_{|X|})}{f_{|X|}(P_{|X|})}\right): f \in S_d\right\}.
$$
\end{definition}

\begin{lemma}{\rm\cite[Lemma~2.13]{hilbert-min-dis}}\label{may21-15-1}
{\rm (a)} The map ${\rm ev}_d$ is independent of
the set of representatives that we choose for the points of
$X$. {\rm (b)} The basic parameters of the Reed-Muller-type code
$C_X(d)$ are independent of $f_1,\ldots,f_{|X|}$.
\end{lemma}

The basic parameters of $C_X(d)$ are related to the algebraic
invariants of the quotient ring $S/I_X$, where $I_X$ is the vanishing
ideal of $X$ (see for example \cite{sarabia4,vaz,algcodes}). Indeed, the
dimension of $C_X(d)$ is given by the Hilbert function of $S/I_X$,
that is, $\dim_K (C_X(d))=H_X(d)$, the length $|X|$ of $C_X(d)$ is
given by the degree, or the multiplicity of $S/I_X$. Moreover, the
regularity index of $H_X$ is the Castelnuovo--Mumford regularity of $S/I_X$.
Recall that the  $a$-invariant of $S/I_X$, denoted $a_X$, is 
the regularity index minus $1$. Unfortunately there is no algebraic
invariant of $S/I_X$ which is equal to the generalized Hamming weight
of $C_X(d)$. However we can express the minimum distance in algebraic
terms, as was recently shown in \cite{hilbert-min-dis}.

\begin{definition} \label{def2} Let $K^*=K \setminus \{0\}$ be the
multiplicative group of $K$. 
The projective torus of $\mathbb{P}^{s-1}$, denoted
$\mathbb{T}_{s-1}$, is given by
$$
\mathbb{T}_{s-1}=\{[t_1: \cdots : t_s] \in \mathbb{P}^{s-1}: t_i \in
K^*\}. 
$$
\end{definition}

Moreover, a projective torus admits a natural generalization as we now
explain. Let $L=K[Z_1,\ldots,Z_n]$ be a polynomial ring over the field $K$
and let $Z^{a_1},\ldots,Z^{a_s}$ be a finite set of mo\-no\-mials.  
As usual if $a_i=(a_{i1},\ldots,a_{in})\in\mathbb{N}^n$, where $\mathbb{N}$ stands for the non-negative integers, 
then we set
\begin{center}
$Z^{a_i}=Z_1^{a_{i1}}\cdots Z_n^{a_{in}} \, \, \mbox{ for } \, \, i=1,\ldots,s$.
\end{center}
\quad Consider the following set
parameterized  by these mo\-no\-mials 
$$
X=\left\{\left[t_1^{a_{11}}\cdots t_n^{a_{1n}}:\cdots:t_1^{a_{s1}}\cdots t_n^{a_{sn}}\right]\in\mathbb{P}^{s-1}	\vert\, t_i\in K^*\right\}.
$$
 
We call $X$ a {\it toric set parameterized by monomials\/} $Z^{a_1},\ldots,Z^{a_s}$, and 
the co\-rres\-pon\-ding Reed-Muller-type code $C_X(d)$ 
is called a {\it code parameterized by mo\-no\-mials\/} or simply a {\it parameterized
code\/}. In this situation $X$ is a subgroup of
the projective torus $\mathbb{T}_{s-1}$. Recall that the basic parameters of
$C_X(d)$ are independent of $f_1,\ldots,f_{|X|}$ (Lemma~\ref{may21-15-1}). Since all entries of
$P_1,\ldots,P_{|X|}$ are non-zero, we can set $f_i=X_1^d$ for $i=1,\ldots,|X|$ in
Definition \ref{def1} to obtain a simple expression for $C_X(d)$. If the monomials in the
pa\-ra\-me\-te\-ri\-za\-tion are given by the edges of a
simple graph $\mathcal{G}$ \cite{harary}, that is,
$Z^{a_1},\ldots,Z^{a_s}$ is the set of all monomials $Z_iZ_j$ such
that $\{Z_i,Z_j\}$ is an edge of $\mathcal{G}$, we say that $C_X(d)$ is parameterized by
the edges of $\mathcal{G}$ (see \cite[Definitions (2) and
(4)]{sarabia6}). Note that if we consider the monomials $Z_1, \ldots, Z_s$,
the set $X$ parameterized by them is precisely the projective torus $\mathbb{T}_{s-1}$.        

One of the important cases studied here is when $X$ is a complete
intersection. For convenience we recall this notion.
\begin{definition}
A set $X \subset {\mathbb{P}^{s-1}}$ is called a (zero--dimensional
ideal--theoretic) complete intersection if the vanishing ideal $I_X$
of $X$ is generated by a regular sequence of $s-1$ elements.
\end{definition}

Next we define the generalized Hamming
weights introduced in \cite{helleseth,klove,wei}, also known as
higher weights, 
effective lengths or Wei weights.   
\begin{definition}
If $\mathcal{B}$ is subset of $K^{|X|}$, the support of this set is 
$$
{\mbox{supp}} \, (\mathcal{B})=\{i: {\mbox{there exists}} \,\, (w_1,\ldots,w_{|X|}) \in \mathcal{B} \,\, {\mbox{such that}} \,\, w_i \neq 0 \}.
$$
The $r$th generalized Hamming weight of the code $C_X(d)$ is given by
$$
d_r(C_X(d))=\min \, \{|{\mbox{supp}} \, (\mathcal{D})|: \mathcal{D} \,\, {\mbox{is a subcode of}} \,\, C_X(d) \,\, {\mbox{and}}\,\, \dim_K \mathcal{D}=r\}, 
$$
for $1 \leq r \leq H_X(d)$. The weight hierarchy of the code $C_X(d)$ is the set of integers $\{d_r(C_X(d)):1 \leq r \leq H_X(d)\}$.
\end{definition}

It is an straightforward fact that $d_1(C_X(d))$ is precisely the
minimum distance of the evaluation code $C_X(d)$. We say
that $C_X(d)$ is an $r$--MDS code if the Singleton-type bound (see
\cite[Corollary 3.1]{tsfasman}) is attained, that is   
$$
d_r(C_X(d))=|X|-H_X(d)+r.
$$

The following notion of a non-degenerate code will play a role here.
\begin{definition}
If $C \subset K^{|X|}$ is a linear code and $\pi_i$ is the $i$-th
projection map 
$$
\pi_i: C \rightarrow K,\ \ \ (v_1,\ldots,v_{|X|})\mapsto v_i,
$$
for $i=1,\ldots,|X|$, we say that $C$ is degenerate if for some $i$
the image of $\pi_i$ 
is zero, otherwise it is called non-degenerate.  
\end{definition}

If $C$ is
non-degenerate then $d_k(C)=|X|$, where $k$ is the dimension of $C$
as a linear subspace of $K^{|X|}$. 

\section{Evaluation codes associated to complete intersections} \label{intersection}
\begin{lemma} \label{lemma1}
Let $X \subset \mathbb{P}^{s-1}$ be a complete intersection and let 
$a_X$ be the $a$--invariant of $S/I_X$. Then
\begin{enumerate}
\item $d_1(C_X(a_X))=2$.
\item $H_X(d)+H_X(a_X-d)=|X|$.
\end{enumerate}
\end{lemma}

\begin{proof}
\cite[Proposition 2.7]{duursma} and \cite[Lemma 3]{sarabia7}. 
\end{proof} 

\begin{proposition} \label{proposition1}
Let $X \subset \mathbb{P}^{s-1}$ be a complete intersection and $a_X$ be the $a$--invariant of $S/I_X$. Then the $r$th generalized Hamming weight of $C_X(a_X)$ is
$$
d_r(C_X(a_X))=r+1,
$$
for $r=1,\ldots,H_X(a_X)=|X|-1$.
\end{proposition}

\begin{proof}
We set $d_r:=d_r(C_X(a_X))$. The claim follows immediately from Lemma \ref{lemma1} and the fact that
$$
2=d_1<d_2<\cdots<d_{|X|-1} \leq |X|. 
$$ 
\end{proof}

\begin{remark} \label{remark1}
Let $d \geq a_X+1$. In these cases $C_X(d)=K^{|X|}$ and therefore $d_r(C_X(d))=r$ for all $r=1,\ldots,H_X(d)=|X|$.
\end{remark}

\begin{theorem} \label{theorem1}
Let $X$ be a complete intersection and let $a_X$ be its
$a$--invariant. Suppose that the codes $C_X(d)$ and $C_X(a_X-d)$ are
non-degenerate for all $1 \leq d <a_X$. Moreover let
$\beta:=H_X(d)$, $\beta \geq d_1(C_X(a_X-d))$. Then  
$$
d_{\beta-i}(C_X(d))=|X|-i,
$$
for $i=0,\ldots,d_1(C_X(a_X-d))-2$.
\end{theorem}

\begin{proof}
Notice that $d_{\beta}(C_X(d))=|X|$, because $C_X(d)$ is non-degenerate. Moreover $C_X(d)^{\perp}$, the dual code of $C_X(d)$, and $C_X(a_X-d)$ have the same weight hierarchy (see \cite[Theorem 2]{sarabia7}).
By using \cite[Corollary 4.1]{tsfasman}, we obtain that $C_X(d)$ is an $r$--MDS code if $r$ is given by
\begin{align*}
r & = |X|+2-H_X(a_X-d)-d_1(C_X(a_X-d)) \\
& =|X|+2-(|X|-\beta)-d_1(C_X(a_X-d))  \\
& =\beta+2-d_1(C_X(a_X-d)). 
\end{align*}

Actually for this $r$ we get
\begin{equation} \label{equa1}
d_r(C_X(d))=|X|+2-d_1(C_X(a_X-d)).
\end{equation}

As $C_X(d)$ is $s$--MDS for any $s\geq r$ (see \cite[Corollary 3.2]{tsfasman}), we conclude that
\begin{equation} \label{equa2}
d_{r+j}(C_X(d))=d_r(C_X(d))+j,
\end{equation}
for $j=0,\ldots,\beta-r$. If we take $i:=\beta-r-j$, we notice that $0\leq i \leq\beta-r=d_1(C_X(a_X-d))-2$. By using Equations (\ref{equa1}) and (\ref{equa2}) we get
$$
d_{\beta-i}(C_X(d))=d_r(C_X(d))+\beta-r-i=|X|-i,
$$
and the claim follows.
\end{proof}

\begin{corollary} \label{corollary1}
With the same notation of Theorem \ref{theorem1}. Let $\alpha=H_X(d)-d_1(C_X(a_X-d))+2$. Then $C_X(d)$ is an $r$--MDS code for $\alpha \leq r \leq |X|$. In fact
\begin{equation} \label{GHWCI}
d_{\alpha+j}(C_X(d))=|X|-d_1(C_X(a_X-d))+2+j,
\end{equation}
for all $j=0,\ldots,d_1(C_X(a_X-d))-2$.
\end{corollary}

\begin{proof}
The claim follows immediately from Equations (\ref{equa1}) and (\ref{equa2}) in the proof of Theorem \ref{theorem1}.
\end{proof}

\begin{example} \label{example2}
We use the notation given in \cite[Example 4.4]{hiram}. Let
$K=\mathbb{F}_{181}$ be a finite field with $181$ elements. Let $X$
be the following projective degenerate torus, which is a complete
intersection (see \cite[Theorem 1]{sarabia5}):    
$$
X=\left\{\left[1:t_1^{90}:t_2^{36}:t_3^{20} \right] \in \mathbb{P}^{3} : t_1,t_2,t_3 \in K^*\right\}.
$$
\quad Let $\beta$ be a generator of the cyclic group $K^*$ and let
$H_1,H_2,H_3$ be the cyclic groups generated by $\beta^{90}$,
$\beta^{36}$, $\beta^{20}$, respectively. 
Note that $X$ is the image of $H_1\times H_2\times H_3$ under
the map $x\mapsto [1:x]$. 

Thus $|X|=|H_1||H_2||H_3|$, and consequently
$$|X|=\frac{180}{\gcd(180,90)} \times \frac{180}{\gcd(180,36)} \times
\frac{180}{\gcd(180,20)}=2 \times 5 \times 9=90.$$  
\quad In this case the $a$-invariant of $S/I_X$ is $a_X=12$ 
and the regularity index of $H_X$ is $13$. 
If  $d=5$, $H_X(5)=35$ and $d_1(C_X(7))=7$, we 
get the last six generalized Hamming weights of $C_X(5)$:
$$
d_{35-i}(C_X(5))=90-i \,\,\, {\mbox{for all}} \,\,\, i=0,\ldots,5,
$$
and $C_{X}(5)$ is $r$--MDS for $30 \leq r \leq 35$.
In a similar way, if $d=7$, $H_X(C_X(7))=55$ and $d_1(C_X(5))=9$, we obtain the last eight generalized Hamming weights of $C_X(7)$:
$$
d_{55-i}(C_X(7))=90-i \,\,\, {\mbox{for all}} \,\,\, i=0,\ldots,7,
$$
and $C_{X}(7)$ is $r$--MDS for $48 \leq r \leq 55$.

\end{example}

\section{Codes parameterized by a set of monomials} \label{monomials}
Let $X\subset {\mathbb{T}_{s-1}}$ be a toric set parameterized by a
set of monomials and let $C_X(d)$ be its associated Reed--Muller-type
code. The following
result generalizes \cite[Proposition 5.2]{algcodes} and shows that
the generalized Hamming weights have the opposite behavior to the Hilbert function. 

\begin{theorem} \label{theorem2}
Let $C_X(d)$ be a parameterized code and let $1\leq r \leq H_X(d)$.
Then $C_X(d)$ is non-degenerate, and if $d_r(C_X(d))=r$, then
$d_r(C_X(d+1))=r$. Moreover, if $d_r(C_X(d))>r$, then
$d_r(C_X(d)) > d_r(C_X(d+1))$.   
\end{theorem}

\begin{proof}
By taking $f=X_1^d \in S_d$, we conclude that $(1,1,\ldots,1) \in C_X(d)$, and then this code is non-degenerate. Let $X=\{P_1,\ldots,P_{|X|}\}$ and $\Lambda_f \in C_X(d)$ with
\begin{equation} \label{Lambda}
\Lambda_f:=\left( \frac{f(P_1)}{X_1^d(P_1)}, \ldots, \frac{f(P_{|X|})}{X_1^d(P_{|X|})}\right),
\end{equation}
where $f \in S_d$. Thus
$$
\Lambda_f=\left( \frac{(X_1 f)(P_1)}{X_1^{d+1}(P_1)}, \ldots, \frac{(X_1 f)(P_{|X|})}{X_1^{d+1}(P_{|X|})}\right),
$$
and then $\Lambda_f \in C_X(d+1)$. This shows that $C_X(d) \subset C_X(d+1)$ for all $d$. Therefore
\begin{equation} \label{eqa}
d_r(C_X(d+1)) \leq d_r(C_X(d)),
\end{equation}
for all $r=1,\ldots,H_X(d)$. If $d_r(C_X(d))=r$ then, by inequality (\ref{eqa}), we get
$$
r \leq d_r(C_X(d+1)) \leq d_r(C_X(d))=r,
$$
and the claim follows in this case.
Now, we consider $d_r(C_X(d)) >r$. Let $\mathcal{D}$ be a subspace of $C_X(d)$ with dimension $r$ and such that $d_r(C_X(d)=|\,{\mbox{supp}}\, (\mathcal{D})|$. Let $\mathcal{B}=\{\Lambda_{f_1},\ldots, \Lambda_{f_r}\}$ be a basis of $\mathcal{D}$ where we use the notation of the Equation (\ref{Lambda}) to $\Lambda_{f_i}$. Therefore
$$
r<d_r(C_X(d))=|\,{\mbox{supp}}\, (\mathcal{D})|=|\,{\mbox{supp}}\, (\mathcal{B})|.
$$

Let $i,j \in {\mbox{supp}} \, (\mathcal{B})$ with $i \neq j$. Then there exists $f_{i_1}, f_{i_2} \in S_d$ (not necessarily different) such that $f_{i_1}(P_i) \neq 0$ and $f_{i_2}(P_j) \neq 0$, where $i_1,i_2 \in \{1,\ldots,r\}$. As $P_i \neq P_j$ we can take, without loss of generality, 
$$P_i=[(1,a_{2i},\ldots,a_{si})], \,\,\,\, P_j=[(1,b_{2j},\ldots,b_{sj})],$$ 
and $a_{ki} \neq b_{kj}$ for some $k \in \{2,\ldots,s\}$. We define the following homogeneous polynomials $g_m \in S_{d+1}$:
$$
g_m:=(a_{ki}X_1-X_k) f_m,
$$
for all $m=1.\ldots,r$. Notice that $g_m(P_i)=0$ for all $m=1,\ldots,r$, but $g_{i_2}(P_j) \neq 0$. Furthermore, let $\mathcal{B}'=\{\Lambda_{g_1},\ldots,\Lambda_{g_r}\}$. $\mathcal{B}'$ is a linearly independent set (because $\mathcal{B}$ is also a linearly independent set) and if $\mathcal{D}'$ is the subspace of $C_X(d+1)$ generated by $\mathcal{B}'$, we conclude that (because $i \in {\mbox{supp}} \, (\mathcal{B}) \setminus {\mbox{supp}} \, (\mathcal{B}')$)
$$
d_r(C_X(d+1)) \leq |\, {\mbox{supp}} \, (\mathcal{D}')|=|\, {\mbox{supp}} \, (\mathcal{B}')| 
<|\, {\mbox{supp}} \, (\mathcal{B})|=d_r(C_X(d)),
$$
and the claim follows.
\end{proof}

\begin{remark}
The behavior of the generalized Hamming weights given in Theo\-rem
\ref{theorem2} is illustrated in the Tables \ref{GHWT2} and \ref{GHWT3} of Example \ref{example}. 
\end{remark}

\section{Codes parameterized by the projective torus} \label{finalsection}
If $X$ is the projective torus $\mathbb{T}_{s-1}$, this is a complete
intersection, the length, dimension, and minimum distance of $C_X(d)$
are known, and so is the $a$--invariant of $S/I_X$ 
(see \cite{duursma}, \cite{sarabia0}, and \cite{ci-codes}).
Furthermore, when $s=2$ the code $C_{\mathbb{T}_1}$ is MDS and its
complete weight hierarchy is given in \cite{GHW2014}. Actually,
Theorem \ref{theorem1} can be used to find some generalized Hamming
weights for the codes $C_{\mathbb{T}_{s-1}}(d)$, as Example
\ref{example} shows. In order to prove the theorem that gives the
second generalized Hamming weight of $C_{\mathbb{T}_{s-1}}(d)$ we use
the following notation: If $f \in S_d$ then        
\begin{equation} \label{ZeroesInitial}
Z_{\mathbb{T}_{s-1}}(f):=\{[P] \in \mathbb{T}_{s-1}: f(P)=0\}.
\end{equation}

Also, from now on we use $\beta$ as a generator of the cyclic group $(K^*,\cdot)$. Notice that if $q=2$ then $|\mathbb{T}_{s-1}|=1$ and $C_{\mathbb{T}_{s-1}}(d)=K$. Thus in this section we assume $q \geq 3$.

\begin{lemma} \label{lemma2}
Let $s=3$ and let $d \in \mathbb{N}$, $1 \leq d \leq q-2$. Then we can find $f_1, f_2 \in S_d$ such that
$$
|Z_{\mathbb{T}_2}(f_1) \cap Z_{\mathbb{T}_2}(f_2)|=1+(d-1)(q-1).
$$
\end{lemma}

\begin{proof}
We notice that
$
\mathbb{T}_2=\{[1:t_1:t_2] \in \mathbb{P}^2:t_1,t_2 \in K^*\}
$.
If $d=1$ then we take $f_1=X_1-X_3$ and $f_2=X_2-X_3$. Therefore
$$
Z_{\mathbb{T}_2}(f_1)=\{[1:\alpha:1] \in \mathbb{P}^2: \alpha \in K^*\}, \,\,
Z_{\mathbb{T}_2}(f_2)=\{[1:\alpha:\alpha] \in \mathbb{P}^2: \alpha \in K^*\}.
$$

Thus
$
|Z_{\mathbb{T}_2}(f_1) \cap Z_{\mathbb{T}_2}(f_2)|=1
$,
and the claim follows for $d=1$. Now let $2 \leq d \leq q-2$. Furthermore, let
$
f_1=(\beta^{q-2}X_1-X_3) \prod_{j=1}^{d-1}(\beta^jX_1-X_3) 
$
and
$
f_2=(X_2-X_3) \prod_{j=1}^{d-1}(\beta^jX_1-X_3)
$.
Notice that $f_1, f_2 \in S_d$ and if $[P]=[1:t_1:t_2] \in \mathbb{T}_2$ then
$
f_1(P)=(\beta^{q-2}-t_2) \prod_{j=1}^{d-1}(\beta^j-t_2)
$
and
$
f_2(P)=(t_1-t_2) \prod_{j=1}^{d-1}(\beta^j-t_2)
$.
Moreover if $B_j:=\{[1:\alpha:\beta^j] : \alpha \in K^*\}$ for $1 \leq j \leq d-1$ then 
$$
Z_{\mathbb{T}_{2}}(f_1)=\{[1:\alpha:\beta^{q-2}]: \alpha \in K^*\} \cup (\cup_{j=1}^{d-1} B_j),
$$
and
$$
Z_{\mathbb{T}_{2}}(f_2)=\{[1:\alpha:\alpha]: \alpha \in K^*\} \cup (\cup_{j=1}^{d-1} B_j).
$$

As $B_{j_1} \cap B_{j_2}= \emptyset$ if $j_1 \neq j_2$,
$
\{[1:\alpha:\beta^{q-2}]: \alpha \in K^*\} \cap \{[1:\alpha:\alpha]: \alpha \in K^*\}=\{[1:\beta^{q-2}:\beta^{q-2}]\},
$
and $[1:\beta^{q-2}:\beta^{q-2}] \notin \cup_{j=1}^{d-1} B_j$, we conclude that
$$
|Z_{\mathbb{T}_2}(f_1) \cap Z_{\mathbb{T}_2}(f_2)|=1+(d-1)(q-1),
$$
and the claim follows.
\end{proof}

\begin{theorem} \label{theoremzeroes}
Let $s \geq 3$, $\eta=(q-2)(s-2)$ and $r=(q-2)(s-1)$. Then we can find $F, G \in S_d$ such that
\begin{align*}
& |Z_{\mathbb{T}_{s-1}}(F) \cap Z_{\mathbb{T}_{s-1}}(G)|    
= \\
& \left\{
\begin{array}{ll}
(q-1)^{s-(k+3)}[(q-1)^{k+2}-(q-1)(q-l)+1]  & 1 \leq d \leq \eta \\
(q-1)^{s-1}-q+l & \eta < d \leq r, 
\end{array} \right.
\end{align*}
where $k$ and $l$ are the unique integers such that $d=k(q-2)+l$, $k \geq 0$  and $1 \leq l \leq q-2$.
\end{theorem}

\begin{proof}
{\bf{Case I:}} Let $1 \leq d \leq q-2$. Let $F(X_1,\ldots,X_s):=f_1(X_1,X_2,X_3)$ and $G(X_1,\ldots,X_s):=f_2(X_1,X_2,X_3)$, where $f_1$ and $f_2$ are the polynomials given in the Lemma \ref{lemma2}. It is easy to see that $[t_1:t_2:t_3] \in Z_{\mathbb{T}_2}(f_1) \cap Z_{\mathbb{T}_2}(f_2)$ if and only if $[t_1:t_2:t_3:t_4:\cdots:t_s] \in Z_{\mathbb{T}_{s-1}}(F) \cap Z_{\mathbb{T}_{s-1}}(G)$ for any $t_4, \ldots, t_s \in K^*$. Therefore
$$
|Z_{\mathbb{T}_{s-1}}(F) \cap Z_{\mathbb{T}_{s-1}}(G)| =(q-1)^{s-3}[1+(d-1)(q-1)],
$$
and due to the fact that in this case $l=d$ and $k=0$, the claim follows.

{\bf{Case II:}} Let $q-2<d \leq (q-2)(s-2)$. Notice that in this case
$1 \leq k \leq s-3$. We define the following polynomials:
$H_k:=\prod_{j=1}^{k}[\prod_{i=1}^{q-2}(\beta^iX_1-X_{j+1})]$,
$f_{k,l}:=(\beta^{q-2}X_1-X_{k+2})
\prod_{i=1}^{l-1}(\beta^iX_1-X_{k+2})$, $g_{k,l}:=(X_{k+2}-X_{k+3})
\prod_{i=1}^{l-1}(\beta^iX_1-X_{k+2})$. Moreover, we take $F:=H_k
\cdot f_{k,l}$ and $G:=H_k \cdot g_{k,l}$. We notice that $F, G \in
S_d$, where $d=k(q-2)+l$. Obviously $Z_{\mathbb{T}_{s-1}}(H_k)
\subset Z_{\mathbb{T}_{s-1}}(F) \cap Z_{\mathbb{T}_{s-1}}(G)$. Let
$[P]=[1:\alpha_2,\ldots,\alpha_s] \in \mathbb{T}_{s-1}$. If there
exists $i \in \{2,\ldots,k+1\}$ such that $\alpha_i \neq 1$, we can
write $\alpha_i=\beta^r$ with $r \leq q-2$, then
$(\beta^rX_1-X_i)(P)=0$, but $\beta^rX_1-X_i$ is a factor of $H_k$,
thus $[P] \in Z_{\mathbb{T}_{s-1}}(F) \cap Z_{\mathbb{T}_{s-1}}(G)$.
Actually if $$A=\{[1:1: \cdots:1:\alpha_{k+2}: \cdots:\alpha_s]:
\alpha_i \in K^*\}$$ then $\mathbb{T}_{s-1} \setminus A \subset
Z_{\mathbb{T}_{s-1}}(F) \cap Z_{\mathbb{T}_{s-1}}(G)$ and
$|\mathbb{T}_{s-1} \setminus A|=(q-1)^{s-1}-(q-1)^{s-(k+1)}$. Now we
need to find out the number of zeroes of $F$ and $G$ that are in $A$.
If $[P] \in A$ then $H_k(P) \neq 0$; thus we need to examine just the
zeroes of $f_{k,l}$ and $g_{k,l}$ that are in A. But $f_{k,l}$ and
$g_{k,l}$ are of degree $l$, and if we use the proof of Lemma
\ref{lemma2} and the Case I above (we consider the entries
$1,k+2,k+3,\ldots,s$ of the points of $A$, that is $s-k$ entries) we
conclude that                   
$$
|Z_{\mathbb{T}_{s-1}}(f_{k,l}) \cap Z_{\mathbb{T}_{s-1}}(g_{k,l}) \cap A| 
=(q-1)^{s-(k+3)}[(q-1)(l-1)+1]. 
$$

Therefore
\begin{align*}
|Z_{\mathbb{T}_{s-1}}(F) \cap Z_{\mathbb{T}_{s-1}}(G)| & =(q-1)^{s-1}-(q-1)^{s-(k+1)} \\
& +(q-1)^{s-(k+3)}[(q-1)(l-1)+1] \\
& =(q-1)^{s-(k+3)}[(q-1)^{k+2} \\
& -(q-1)^2 +(q-1)(l-1)+1] \\
& =(q-1)^{s-(k+3)}[(q-1)^{k+2}-(q-1)(q-l)+1], 
\end{align*}
and the claim follows.

{\bf{Case III:}} Let $(q-2)(s-2) < d \leq (q-2)(s-1)$. In this case $k=s-2$. We use the polynomials $f_{s-2,l}=(\beta^{q-2}X_1-X_s) \prod_{i=1}^{l-1} (\beta^iX_1-X_s)$ and $g_{s-2,l}=(X_1-X_s) \prod_{i=1}^{l-1} (\beta^iX_1-X_s)$. We continue using the notation introduced above for the remaining polynomials. It is immediate that $Z_{\mathbb{T}_{s-1}}(H_{s-2})=\mathbb{T}_{s-1} \setminus B$, where $B=\{[1:1:\cdots:1:\alpha_s]: \alpha_s \in K^*\}$. Therefore
$$
|Z_{\mathbb{T}_{s-1}}(H_{s-2})|=(q-1)^{s-1}-(q-1).
$$

Moreover
$$
Z_{\mathbb{T}_{s-1}}(f_{s-2,l}) \cap Z_{\mathbb{T}_{s-1}}(g_{s-2,l}) \cap B 
= \{[1:1: \cdots :1:\beta^i] : i=1,\ldots, l-1\}.
$$

Thus $| Z_{\mathbb{T}_{s-1}}(f_{s-2,l}) \cap Z_{\mathbb{T}_{s-1}}(g_{s-2,l}) \cap B |=l-1$. We conclude that
$$
|Z_{\mathbb{T}_{s-1}}(F) \cap Z_{\mathbb{T}_{s-1}} (G)| = 
(q-1)^{s-1}-(q-1)+l-1 
= (q-1)^{s-1}-q+l,
$$
and the claim follows.
\end{proof}

\begin{remark}\label{dec21-17}
The formula for the minimum distance of the codes $C_{\mathbb{T}_{s-1}}(d)$ was found in \cite[Theorem 3.5]{ci-codes}:
\begin{equation} \label{eqmd}
d_1(C_{\mathbb{T}_{s-1}}(d))= \left\{
\begin{array}{lll}
(q-1)^{s-(k+2)}(q-1-l) & {\mbox{if}} & 1 \leq d <r \\
1 & {\mbox{if}} & d\geq r, 
\end{array} \right.
\end{equation}
where $k$ and $l$ are the unique integers such that $d=k(q-2)+l$, $k \geq 0$, $1 \leq l \leq q-2$, and $r=(q-2)(s-1)$. It is easy to see that Equation (\ref{eqmd}) can be reduced to
\begin{equation} \label{eqmd2}
d_1(C_{\mathbb{T}_{s-1}}(d))=\left\lceil (q-1)^{s-(k+2)}(q-1-l)\right\rceil ,
\end{equation}
for all $d \geq 1$.
\end{remark}

From now on we use the following notation:
$$
\mathcal{Z}_1(s,d):=(q-1)^{s-1}-(q-1)^{s-(k+2)}(q-1-l),
$$
where $k$ and $l$ are the unique integers such that $d=k(q-2)+l$, $k \geq 0$  and $1 \leq l \leq q-2$.
It is immediate that if $f \in S_d \setminus I_{\mathbb{T}_{s-1}}(d)$ then
$$
|Z_{\mathbb{T}_{s-1}}(f)| \leq \mathcal{Z}_1(s,d).
$$

Also from now on we use $\mathcal{Z}_2(s,d)$ as the right hand side of the equation involved in Theorem \ref{theoremzeroes}, that is
$$
\mathcal{Z}_2(s,d):=
\left\{
\begin{array}{ll}
(q-1)^{s-(k+3)}[(q-1)^{k+2}-(q-1)(q-l)+1]  & 1 \leq d \leq \eta \\
(q-1)^{s-1}-q+l & \eta < d \leq r, 
\end{array} \right.
$$
where $k$ and $l$ are the unique integers such that $d=k(q-2)+l$, $k \geq 0$  and $1 \leq l \leq q-2$.

\begin{lemma} \label{ineq2}
With the notation introduced above. Let $s, d \in \mathbb{N}$, $s \geq 3$.
\begin{enumerate}
\item If $1 \leq d'<d \leq (q-2)(s-1)$ then
\begin{enumerate}
\item $\mathcal{Z}_1(s,d') \leq \mathcal{Z}_1(s,d)$.
\item $\mathcal{Z}_2(s,d') \leq \mathcal{Z}_2(s,d)$.
\item $\mathcal{Z}_2(s,d') \leq \mathcal{Z}_1(s,d') \leq \mathcal{Z}_2(s,d) \leq \mathcal{Z}_1(s,d)$.
\end{enumerate}

\bigskip
\item \begin{enumerate}
\item $(q-1) \mathcal{Z}_1(s,d) = \mathcal{Z}_1(s+1,d)$.
\item $(q-1) \mathcal{Z}_2(s,d) \leq \mathcal{Z}_2(s+1,d)$.
\end{enumerate}
\end{enumerate}
\end{lemma}

\begin{proof}
\begin{enumerate}
\item \begin{enumerate}
\item This result is an obvious consequence of the fact that $\mathcal{Z}_1(s,d)=(q-1)^{s-1}-d_1(C_{\mathbb{T}_{s-1}}(d))$, and \cite[Proposition 5.2]{algcodes}.
\item Let $2 \leq d \leq (q-2)(s-2)$. If $l >1$ then
\begin{align*}
\mathcal{Z}_2(s,d-1) 
& =(q-1)^{s-1}-(q-1)^{s-(k+2)}(q+1-l) \\
& +(q-1)^{s-(k+3)} 
\leq (q-1)^{s-1} \\
& -(q-1)^{s-(k+2)}(q-l) 
+(q-1)^{s-(k+3)}=\mathcal{Z}_2(s,d).
\end{align*}

If $l=1$ we obtain that
\begin{align*}
\mathcal{Z}_2(s,d-1) 
& =(q-1)^{s-1}-2(q-1)^{s-(k+1)} 
+(q-1)^{s-(k+2)} \\
& = (q-1)^{s-1}-(q-1)^{s-(k+2)}(q-1) 
-(q-1)^{s-(k+1)} \\
& +(q-1)^{s-(k+2)} 
\leq (q-1)^{s-1}-(q-1)^{s-(k+2)}(q-l) \\
& +(q-1)^{s-(k+3)}=\mathcal{Z}_2(s,d).
\end{align*}

Let $(q-2)(s-2) < d \leq (q-2)(s-1)$. If we take $l >1$ then
$$
\mathcal{Z}_2(s,d-1) = (q-1)^{s-1}-q+l-1 
\leq (q-1)^{s-1}-q+l=\mathcal{Z}_2(s,d).
$$

If $l=1$ then
$$
\mathcal{Z}_2(s,d-1) = (q-1)^{s-1}-2(q-1)+1
\leq (q-1)^{s-1}-q+1=\mathcal{Z}_2(s,d),
$$
and the claim follows.
\item It is immediate that for any $d$, $\mathcal{Z}_2(s,d) \leq \mathcal{Z}_1(s,d)$. Let $1 \leq d' <d \leq (q-2)(s-2)$. If $l >1$ then
\begin{align*}
\mathcal{Z}_1(s,d') & \leq \mathcal{Z}_1(s,d-1) =(q-1)^{s-1}-(q-1)^{s-(k+2)}(q-l) \\
& \leq (q-1)^{s-1}-(q-1)^{s-(k+2)}(q-l)
+(q-1)^{s-(k+3)} \\
& =\mathcal{Z}_2(s,d).
\end{align*}

If $l=1$ then
\begin{align*}
\mathcal{Z}_1(s,d') & \leq \mathcal{Z}_1(s,d-1)
=(q-1)^{s-1}-(q-1)^{s-(k+1)}\\
& \leq (q-1)^{s-1}-(q-1)^{s-(k+2)}(q-1) +(q-1)^{s-(k+3)} \\
& =\mathcal{Z}_2(s,d).
\end{align*}

Let $(q-2)(s-2)<d \leq (q-2)(s-1)$ and we suppose that $d'<d$. If $l >1$ then
\begin{align*}
\mathcal{Z}_1(s,d')\leq \mathcal{Z}_1(s,d-1) 
& =(q-1)^{s-1}-(q-1-l+1) \\
& = (q-1)^{s-1}-q+l 
=\mathcal{Z}_2(s,d). 
\end{align*}

If $l=1$ then
$$
\mathcal{Z}_1(s,d') \leq \mathcal{Z}_1(s,d-1) 
=(q-1)^{s-1}-(q-1)=\mathcal{Z}_2(s,d),
$$
and the claim follows.
\end{enumerate}
\item \begin{enumerate}
\item It is an obvious conclusion because of the definition of $\mathcal{Z}_1(s,d)$.
\item Let $1 \leq d \leq (q-2)(s-2)$. Then
\begin{align*}
(q-1) \mathcal{Z}_2(s,d) & = (q-1)^s -(q-1)^{s-(k+1)}(q-l) 
+(q-1)^{s-(k+2)} \\
& =\mathcal{Z}_2(s+1,d).
\end{align*}

If $(q-2)(s-2) < d \leq (q-2)(s-1)$ then
\begin{align*}
(q-1) \mathcal{Z}_2(s,d) & = (q-1)^s-(q-1)(q-l)
\leq (q-1)^s-(q-l) \\
& =\mathcal{Z}_2(s+1,d),
\end{align*}
\end{enumerate}
\end{enumerate}
and the whole claim follows.
\end{proof}
Moreover, we need the following definition: If we take $f \in K[X_1,\ldots,X_s]_d$ and $a \in K^*$ then $$f_a(X_1,\ldots,X_{s-1}):=f(X_1,\ldots,X_{s-1},aX_1).$$ 

Notice that $f_a \in K[X_1,\ldots,X_{s-1}]_d$. We assume that the
points in the projective space are in standard form, that is, the first non-zero entry from the left is $1$. If $[Q] \in \mathbb{P}^{s-2}$, $[Q]=[t_1: \cdots : t_{s-1}]$ we denote by $[Q_a]:=[t_1:\cdots:t_{s-1}:at_1] \in \mathbb{P}^{s-1}$.

\begin{lemma} \label{lemma6}
With the notation introduced above.
\begin{enumerate}
\item $f_a=0$ (the zero polynomial)  if and only if $aX_1-X_s$ divides $f$.
\item Let $P=[1:b_2:\cdots:b_s] \in \mathbb{T}_{s-1}$ and $Q=[1:b_2:\cdots:b_{s-1}] \in \mathbb{T}_{s-2}$ (obviously $[P]=[Q_{b_s}]$). Then
$f(P)=0$ if and only if $f_{b_s}(Q)=0$.
\end{enumerate}
\end{lemma}

\begin{proof}
The second part of the proof is immediate from the definitions. In order to prove the first assertion we suppose that $f_a=0$. Fix a monomial ordering in which $X_s> \cdots >X_1$. By using the division algorithm there are $g, r \in K[X_1,\ldots,X_s]$ such that $f=(aX_1-X_s) g +r$, where $r=0$ or $r$ is a $K$--linear combination of monomials, none of which is divisible by $X_s$. Therefore $r \in K[X_1,\ldots,X_{s-1}]$. But
$$
0=f_a(X_1,\ldots,X_{s-1})=f(X_1,\ldots,X_{s-1},aX_1) 
=r(X_1,\ldots,X_{s-1}).
$$

Thus $r=0$ and $aX_1-X_s$ divides $f$. The converse is obvious.
\end{proof}

\begin{theorem} \label{otherbound}
Let $s, d \in \mathbb{N}$, $1 \leq d \leq (q-2)(s-2)$, $s \geq 3$, and let $f, g$ two linearly independent polynomials on $S_d \setminus I_{\mathbb{T}_{s-1}}(d)$. Then
$$
|Z_{\mathbb{T}_{s-1}}(f) \cap Z_{\mathbb{T}_{s-1}}(g)| \leq (q-1)^{s-(k+3)} [(q-1)^{k+2} 
-(q-1)(q-l)+1],
$$
where $k$ and $l$ are the unique integers such that  $d=k(q-2)+l$, $k \geq 0$, $1 \leq l \leq q-2$.
\end{theorem}

\begin{proof}
Notice that if $f \in K[X_1,\ldots,X_s]_d$,
\begin{equation} \label{form1}
f=\sum_{i=1}^{r_1} (a_iX_1-X_s)^{n_i} f',
\end{equation}
where $a_1,\dots,a_{r_1}$ are different non-zero elements of $K$, $f'_a \neq 0$ for all $a \in K^*$, and
\begin{equation} \label{form2}
\tilde{f}=\sum_{i=1}^{r_1} (a_iX_1-X_s) f',
\end{equation}
then $|Z_{\mathbb{T}_{s-1}}(f)|=|Z_{\mathbb{T}_{s-1}}(\tilde{f})|$. Thus there is no loss of generality if we assume that any polynomial of the form (\ref{form1}) can be studied as if it were of the form (\ref{form2}).
We proceed by induction on $s$ (the number of variables). Let $s=3$ and $f,g$ be two linearly independent polynomials in $K[X_1,X_2,X_3]_d \setminus I_{\mathbb{T}_2}(d)$  with $1 \leq d \leq (3-2)(q-2)=q-2$. For all $a \in K^*$, $f_a, g_a \in K[X_1,X_2]_d$. Let $\mathcal{A}:=\{ a \in K^* : f_a=0\}$, and $\mathcal{B}:=\{b \in K^*: g_b=0\}$. We consider the following cases:

{\bf{Case A:}} $\mathcal{A}=\mathcal{B}=\emptyset$. Thus $f_a \neq 0$, $g_a \neq 0$ for all $a \in K^*$. By using \cite[Equation (6)]{GHW2014} for $r=2$ we obtain that
$$
|Z_{\mathbb{T}_1}(f_a) \cap Z_{\mathbb{T}_1}(g_a)| \leq (q-1)-(q-d)=d-1=l-1.
$$

If $a$ runs over all the elements of $K^*$ and we use Lemma \ref{lemma6} then
$$
|Z_{\mathbb{T}_2}(f) \cap Z_{\mathbb{T}_2}(g)| \leq (q-1)(l-1) \leq (q-1)(l-1)+1,
$$
and the case A follows.

{\bf{Case B:}} $\mathcal{A}=\{a_1,\ldots,a_{r_1}\} \neq \emptyset$, $\mathcal{B}=\{b_1,\ldots,b_{r_2}\} \neq \emptyset$, and $\mathcal{A} \cap \mathcal{B}= \emptyset$. In this case we can write $f=f'H_1$, where $H_1=\prod_{i=1}^{r_1} (a_iX_1-X_3)$ and $f'_a \neq 0$ for all $a \in K^*$. Also $g=g' H_2$ with $H_2=\prod_{i=1}^{r_2} (b_iX_1-X_3)$ and $g'_a \neq 0$ for all $a \in K^*$. As $Z_{\mathbb{T}_{2}}(H_1) \cap Z_{\mathbb{T}_2}(H_2)=\emptyset$, we obtain that
\begin{align} \label{eq7}
|Z_{\mathbb{T}_2}(f) \cap Z_{\mathbb{T}_2}(g)| & =|Z_{\mathbb{T}_2}(H_1) \cap Z_{\mathbb{T}_2}(g')| 
+|Z_{\mathbb{T}_2}(H_2) \cap Z_{\mathbb{T}_2}(f')| \nonumber \\
& +|Z_{\mathbb{T}_2}(f') \cap Z_{\mathbb{T}_2}(g') \setminus Z_{\mathbb{T}_2}(H_1H_2)|. 
\end{align}

Notice that $(H_1)_{a_i}=0$ for all $i=1\ldots,r_1$. As $|Z_{\mathbb{T}_{s-1}}(f)| \leq \mathcal{Z}_1(s,d)$ and $\deg g'_{a_i}=d-r_2$ we get
$$
|Z_{\mathbb{T}_1}(g'_{a_i})| \leq l-r_2.
$$

As $i \in \{1,\ldots,r_1\}$ we obtain that
$$
|Z_{\mathbb{T}_2}(H_1) \cap Z_{\mathbb{T}_2}(g')| \leq r_1(l-r_2).
$$

In the same way
$$
|Z_{\mathbb{T}_2}(H_2) \cap Z_{\mathbb{T}_2}(f')| \leq r_2(l-r_1).
$$

Let $r_3=\min \{r_1,r_2\}$ and consider the polynomials $f'':=X_1^{r_1-r_3} f'$ and $g'':=X_1^{r_2-r_3} g'$. Notice that $Z_{\mathbb{T}_2}(f'')=Z_{\mathbb{T}_2}(f')$, $Z_{\mathbb{T}_2}(g'')=Z_{\mathbb{T}_2}(g')$, and  $\deg f''=\deg g''=d-r_3$. If we proceed similarly to case A then
\begin{align*}
|Z_{\mathbb{T}_2}(f') \cap Z_{\mathbb{T}_2}(g') \setminus Z_{\mathbb{T}_2}(H_1H_2)| 
& \leq (q-1-r_1-r_2)(l-1-r_3) \\
& \leq (q-1-r_1-r_2)(l-r_3),
\end{align*}
and by using Equation (\ref{eq7}) we conclude that
\begin{align*}
 |Z_{\mathbb{T}_2}(f) \cap Z_{\mathbb{T}_2}(g)| & \leq r_1(l-r_2) +r_2(l-r_1)+(q-1-r_1-r_2)(l-r_3) \\
& = -2r_1r_2+(q-1)(l-r_3)+r_1r_3 +r_2r_3 \leq (q-1)(l-r_3) \\
& = (q-1)(l-1)-(q-1)(r_3-1) \leq (q-1)(l-1)+1,
\end{align*}
and the case B follows.

{\bf{Case C:}} $\mathcal{A}=\emptyset$, $\mathcal{B}=\{b_1,\ldots,b_{r_2}\} \neq \emptyset$. Thus $g=g' H_2$, where $g'$ and $H_2$ were defined in case B, $f_a \neq 0$ and $g'_{a} \neq 0$ for all $a \in K^*$. Therefore, as above,
$$
|Z_{\mathbb{T}_2}(H_2) \cap Z_{\mathbb{T}_2}(f)| \leq r_2l.
$$

If we define the polynomial $X_1^{r_2}g'$, then this polynomial has the same zeroes than $g'$ and its degree is $d$. Moreover
$$
|Z_{\mathbb{T}_1}(X_1^{r_2}g'_{b_i}) \cap Z_{\mathbb{T}_1}(f_{b_i})| =| Z_{\mathbb{T}_1}(g'_{b_i}) \cap Z_{\mathbb{T}_1}(f_{b_i})| 
\leq |Z_{\mathbb{T}_1}(g'_{b_i})| \leq l-r_2.
$$

Due to the fact that $i=1,\ldots,r_2$, we get
$$
|Z_{\mathbb{T}_2}(f) \cap Z_{\mathbb{T}_2}(g') \setminus Z_{\mathbb{T}_{2}}(H_2)| \leq (q-1-r_2)(l-r_2),
$$
and then
\begin{align*}
|Z_{\mathbb{T}_2}(f) \cap Z_{\mathbb{T}_2}(g)| & = |Z_{\mathbb{T}_2}(H_2) \cap Z_{\mathbb{T}_2}(f)| 
+|Z_{\mathbb{T}_2}(f) \cap Z_{\mathbb{T}_2}(g') \setminus Z_{\mathbb{T}_{2}}(H_2)|  \\
& \leq r_2l+(q-1-r_2)(l-r_2) 
=(q-1)(l-r_2)+r_2^2 \\
& =(q-1)(l-1)-(r_2-1)(q-1)+r_2^2\\
& =(q-1)(l-1)+1 
-(q-2-r_2)(r_2-1) \\
& \leq (q-1)(l-1)+1,
\end{align*}
because $1 \leq r_2 \leq q-2$, and the case C follows.

{\bf{Case D:}} $\mathcal{A} \cap \mathcal{B}=\{c_1,\ldots,c_{r_4}\} \neq \emptyset$. Thus $f=f'H$, $g=g'H$ with $H=\prod_{i=1}^{r_4}(c_iX_1-X_3)$, and $f'_{c_i} \neq 0$, $g'_{c_i} \neq 0$ for all $i=1,\ldots,r_4$. Notice that $1 \leq r_4 \leq l=d$ and if $r_4=l$ then $f$ and $g$ are linearly dependent polynomials, which is wrong. That is, $r_4<l$. Notice that, by using \cite[Equation (6)]{GHW2014} with $r=2$, 
$$
|Z_{\mathbb{T}_1}(f'_{c_i}) \cap Z_{\mathbb{T}_1}(g'_{c_i})| \leq l-1-r_4.
$$

Therefore
$$
|Z_{\mathbb{T}_2}(f') \cap Z_{\mathbb{T}_2}(g') \setminus Z_{\mathbb{T}_2}(H)| \leq |Z_{\mathbb{T}_2}(f') \cap Z_{\mathbb{T}_2}(g')|  
\leq (q-1-r_4)(l-1-r_4),
$$
and thus
\begin{align*}
|Z_{\mathbb{T}_2}(f) \cap Z_{\mathbb{T}_2}(g)| & =|Z_{\mathbb{T}_2}(H)| 
+|Z_{\mathbb{T}_2}(f') \cap Z_{\mathbb{T}_2}(g') \setminus Z_{\mathbb{T}_2}(H)| \\
& \leq r_4(q-1)+(q-1-r_4)(l-1-r_4) \\
& =(q-1)(l-1)-r_4(l-1-r_4) 
\leq (q-1)(l-1)+1,
\end{align*}
and the case D follows. 

Cases A, B, C, and D prove the claim for $s=3$. We assume that the result follows for $s$ and we will prove it for $s+1$. Let $f$ and $g$ be two linearly independent polynomials in $K[X_1,\ldots,X_{s+1}]_d \setminus I_{\mathbb{T}_s}(d)$. We continue using the notation for the sets $\mathcal{A}$ and $\mathcal{B}$ introduced above. Although this proof is quite similar to the case $s=3$, there are some details that must be explained. We divide the proof into the following cases.

{\bf{Case I:}} $\mathcal{A}=\mathcal{B}=\emptyset$. Thus $f_a \neq 0$, $g_a \neq 0$ for all $a \in K^*$. By the inductive hypothesis we know that
$$
|Z_{\mathbb{T}_{s-1}}(f_a) \cap Z_{\mathbb{T}_{s-1}}(g_a)| \leq \mathcal{Z}_2(s,d).
$$

Therefore, by using Lemmas \ref{ineq2} and \ref{lemma6} we get
$$
|Z_{\mathbb{T}_{s}}(f) \cap Z_{\mathbb{T}_{s}}(g)| \leq (q-1) \mathcal{Z}_2(s,d) \leq \mathcal{Z}_2(s+1,d),
$$
and the case I follows.

{\bf{Case II:}} $\mathcal{A}=\{a_1,\ldots,a_{r_1}\} \neq \emptyset$, $\mathcal{B}=\{b_1,\ldots,b_{r_2}\} \neq \emptyset$, and $\mathcal{A} \cap \mathcal{B}= \emptyset$. In this case we can write $f=f'H_1$, where $H_1=\prod_{i=1}^{r_1} (a_iX_1-X_{s+1})$ and $f'_a \neq 0$ for all $a \in K^*$. Also $g=g' H_2$ with $H_2=\prod_{i=1}^{r_2} (b_iX_1-X_{s+1})$ and $g'_a \neq 0$ for all $a \in K^*$. We observe that  $f' \in K[X_1,\ldots,X_{s+1}]_{d-r_1}$, and $g' \in K[X_1,\ldots,X_{s+1}]_{d-r_2}$. As $Z_{\mathbb{T}_{s}}(H_1) \cap Z_{\mathbb{T}_s}(H_2)=\emptyset$, we obtain that
\begin{align*}
|Z_{\mathbb{T}_{s}}(f) \cap Z_{\mathbb{T}_{s}}(g)| & = |Z_{\mathbb{T}_{s}}(H_2) \cap Z_{\mathbb{T}_{s}}(f')| 
+|Z_{\mathbb{T}_{s}}(H_1) \cap Z_{\mathbb{T}_{s}}(g')| \\
& + |Z_{\mathbb{T}_{s}}(f') \cap Z_{\mathbb{T}_{s}}(g') \setminus Z_{\mathbb{T}_s}(H_1H_2)|.
\end{align*}

Also, by the definition of $\mathcal{Z}_1(s,d)$, we get
$$
|Z_{\mathbb{T}_{s-1}}(g'_a)| \leq \mathcal{Z}_1(s,d-r_2),
$$
for all $i=1,\ldots,r_1$. By using the fact that $(H_1)_{a_i}=0$ for all $i=1,\ldots,r_1$ and Lemma \ref{ineq2}, we obtain that
\begin{equation} \label{l1}
|Z_{\mathbb{T}_{s}}(g') \cap Z_{\mathbb{T}_{s}}(H_1)| \leq r_1 \mathcal{Z}_1(s,d-r_2) \leq r_1 \mathcal{Z}_2(s,d).
\end{equation}

In exactly the same way we get
\begin{equation} \label{l2}
|Z_{\mathbb{T}_{s}}(f') \cap Z_{\mathbb{T}_{s}}(H_2)| \leq r_2 \mathcal{Z}_1(s,d-r_1) \leq r_2 \mathcal{Z}_2(s,d).
\end{equation}

Let $r_3=\min \{r_1,r_2\}$ and consider the polynomials $f'':=X_1^{r_1-r_3} f'$ and $g'':=X_1^{r_2-r_3} g'$. Notice that $Z_{\mathbb{T}_{s}}(f'')=Z_{\mathbb{T}_{s}}(f')$, $Z_{\mathbb{T}_{s}}(g'')=Z_{\mathbb{T}_{s}}(g')$, and  $\deg f''=\deg g''=d-r_3$. Thus, by the inductive hypothesis,
$$
|Z_{\mathbb{T}_{s-1}}(f'_a) \cap Z_{\mathbb{T}_{s-1}}(g'_a)|=|Z_{\mathbb{T}_{s-1}}(f''_a) \cap Z_{\mathbb{T}_{s-1}}(g''_a)| 
\leq \mathcal{Z}_2(s,d-r_3).
$$

As $|K^*-(\mathcal{A} \cup \mathcal{B})|=q-1-r_1-r_2$, then
\begin{align} \label{l3}
|Z_{\mathbb{T}_{s}}(f') \cap Z_{\mathbb{T}_{s}}(g' )\setminus Z_{\mathbb{T}_{s}}(H_1H_2)| 
& \leq (q-1-r_1-r_2) \mathcal{Z}_2(s,d-r_3) \nonumber \\
& \leq (q-1-r_1-r_2) \mathcal{Z}_2(s,d). 
\end{align}

By using Equations (\ref{l1}), (\ref{l2}), (\ref{l3}) and Lemma \ref{ineq2}, we conclude that
$$
|Z_{\mathbb{T}_{s}}(f) \cap Z_{\mathbb{T}_{s}}(g)| \leq (q-1) \mathcal{Z}_2(s,d) \leq \mathcal{Z}_2(s+1,d),
$$
and the case II follows.

{\bf{Case III:}} $\mathcal{A}=\emptyset$, $\mathcal{B}=\{b_1,\ldots,b_{r_2}\} \neq \emptyset$. Thus $g=g' H_2$, where $g'$ and $H_2$ were defined in case II, $f_a \neq 0$ for all $a \in K^*$, and $g'_{b_i} \neq 0$ for all $i=1,\ldots,r_2$. Similar to the previous cases we get
$$
|Z_{\mathbb{T}_{s}}(f) \cap Z_{\mathbb{T}_{s}}(g)| \leq r_2 \mathcal{Z}_1(s,d)+(q-1-r_2) \mathcal{Z}_1(s,d-r_2).
$$

If $r_2 <l$ then
\begin{align*}
& r_2 \mathcal{Z}_1(s,d)+(q-1-r_2) \mathcal{Z}_1(s,d-r_2)  \\
&= \mathcal{Z}_2(s+1,d)-(q-1)^{s-(k+2)}(q-2-r_2)(r_2-1) 
\leq \mathcal{Z}_2(s+1,d).
\end{align*}

If $r_2 \geq l$, as $r_2=|\mathcal{B}| \leq |K^*|=q-1$, then $1 \leq l \leq r_2 \leq q-1$. If $r_2=q-1$, then $g_a=0$ for all $a \in K^*$. Thus $g \in I_{\mathbb{T}_{s}}(d)$, which is false. Therefore $1 \leq l \leq r_2 \leq q-2$. Moreover
\begin{align*}
& r_2 \mathcal{Z}_1(s,d)+(q-1-r_2) \mathcal{Z}_1(s,d-r_2) 
= \mathcal{Z}_2(s+1,d) \\
& -(q-1)^{s-(k+2)}[r_2(q-1-l)+1] \\
& -(q-1)^{s-(k+1)}[(q-1-r_2)(r_2+1-l)-(q-l)] \\
& =\mathcal{Z}_2(s+1,d)-(q-1)^{s-(k+2)}[r_2(q-1-l)+1] \\
& -(q-1)^{s-(k+1)}[(q-2-r_2)(r_2-l)-1] 
=\mathcal{Z}_2(s+1,d) \\
& -(q-1)^{s-(k+2)}[(q-1)[(q-2-r_2) 
(r_2-l)-1]+r_2(q-1-l)+1] \\
& \leq \mathcal{Z}_2(s+1,d)-(q-1)^{s-(k+2)}(l-1)(q-2-r_2) 
\leq \mathcal{Z}_2(s+1,d),
\end{align*}
and the case III follows.

{\bf{Case IV:}}  $\mathcal{A} \cap \mathcal{B}=\{c_1,\ldots,c_{r_4}\} \neq \emptyset$. Thus $f=f'H$, $g=g'H$ with $H=\prod_{i=1}^{r_4}(c_iX_1-X_3)$, and $f'_{c_i} \neq 0$, $g'_{c_i} \neq 0$ for all $i=1,\ldots,r_4$. We know that
$$
|Z_{\mathbb{T}_{s}}(f) \cap Z_{\mathbb{T}_{s}}(g)| = |Z_{\mathbb{T}_{s}}(H)|+|Z_{\mathbb{T}_{s}}(f') \cap Z_{\mathbb{T}_{s}}(g') \setminus Z_{\mathbb{T}_{s}}(H)|.
$$

By the inductive hypothesis,
$$
|Z_{\mathbb{T}_{s-1}}(f'_{c_i}) \cap Z_{\mathbb{T}_{s-1}}(g'_{c_i})| \leq \mathcal{Z}_2(s,d-r_4),
$$
for all $i=1,\ldots,r_4$. Therefore
$$
|Z_{\mathbb{T}_{s}}(f') \cap Z_{\mathbb{T}_{s}}(g') \setminus Z_{\mathbb{T}_{s}}(H)| \leq (q-1-r_4) \mathcal{Z}_2(s,d-r_4).
$$

Also, as $H_{c_i}=0$ for all $i=1,\ldots,r_4$, we get that $|Z_{\mathbb{T}_{s}}(H)| \leq r_4(q-1)^{s-1}$, and thus
$$
|Z_{\mathbb{T}_{s}}(f) \cap Z_{\mathbb{T}_{s}}(g)| \leq r_4(q-1)^{s-1}+(q-1-r_4) \mathcal{Z}_2(s,d-r_4).
$$

If $r_4 <l$ then
\begin{align*}
& r_4(q-1)^{s-1}+(q-1-r_4) \mathcal{Z}_2(s,d-r_4) \\
& = (q-1) \mathcal{Z}_2(s,d)-r_4(q-1)^{s-(k+3)}[(q-1)^2 
-(q-1)(q+r_4-l)+1] \\
& \leq (q-1) \mathcal{Z}_2(s,d)-r_4(q-1)^{s-(k+3)} \leq \mathcal{Z}_2(s+1,d).
\end{align*}

If $r_4 \geq l$, $d=k(q-2)+l$, then $d-r_4=(k-1)(q-2)+q-2+l-r_4$. Therefore it is easy to see that
\begin{align*}
& r_4(q-1)^{s-1}+(q-1-r_4) \mathcal{Z}_2(s,d-r_4) 
=(q-1)^s \\
& -(q-1)^{s-(k+1)}(r_4+2-l)(q-1-r_4) 
+ (q-1-r_4)(q-1)^{s-(k+2)} \\
& =\mathcal{Z}_2(s+1,d) 
+ (q-1)^{s-(k+2)}[(q-1)(q-l) \\
& -(q-1)(r_4+2-l) 
\cdot (q-1-r_4)+q-2-r_4] \leq \mathcal{Z}_2(s+1,d).
\end{align*}

Thus case IV follows and so does the claim.
\end{proof}

\begin{remark}
If $s=2$ and $q=3$ then (see \cite {GHW2014}) $d_2(C_{\mathbb{T}_{1}}(d))=2$ for all $d \geq 1$. Moreover if $q>3$ the second generalized Hamming weight of $C_{\mathbb{T}_{1}}(d)$ is given by (see \cite[Equation (6)]{GHW2014})
\begin{equation} \label{eqshw3}
d_2(C_{\mathbb{T}_1}(d))= \left\{
\begin{array}{lll}
q-d & {\mbox{if}} & 1 \leq d \leq q-3, \\
2 & {\mbox{if}} & d>q-3.
\end{array} \right.
\end{equation}
\end{remark}

\begin{theorem} \label{theorem3}
The second generalized Hamming weight of the code $C_{\mathbb{T}_{s-1}}(d)$, $s \geq 3$, $d \geq 1$ is given by
\begin{equation} \label{eqshw}
d_2(C_{\mathbb{T}_{s-1}}(d))= 
\left\{
\begin{array}{lll}
(q-1)^{s-(k+3)}[(q-1)(q-l)-1] & {\mbox{if}} & 1 \leq d \leq \eta, \\
q-l & {\mbox{if}} & \eta < d < r, \\
2 & \mbox{if} & d \geq r,
\end{array} \right. 
\end{equation}
where $k$ and $l$ are the unique integers such that  $d=k(q-2)+l$, $k \geq 0$, $1 \leq l \leq q-2$, $\eta=(q-2)(s-2)$ and $r=(q-2)(s-1)$.
\end{theorem}

\begin{proof} 
{\bf{Case I:}} $d \geq (q-2)(s-1)$. As the regularity index in this case is $(q-2)(s-1)$, then $C_{\mathbb{T}_{s-1}}(d)=K^{(q-1)^{s-1}}$. Therefore $d_2(C_{\mathbb{T}_{s-1}}(d))=2$, and the claim follows.

{\bf{Case II:}} $(q-2)(s-2) < d <(q-2)(s-1)$. In this case $k=s-2$. Let $F$ and $G$ be the polynomials given in the Case III of the proof of Theorem \ref{theoremzeroes}. Clearly, the corresponding codewords $\Lambda_F$ and $\Lambda_G$ (we use the notation given in the proof of Theorem \ref{theorem2} with $X=\mathbb{T}_{s-1}$) are linearly independent (because it is easy to find $[P] \in \mathbb{T}_{s-1}$ such that $f_{s-2,l}(P)=0$, but $g_{s-2,l}(P) \neq 0$, for example $[P]=[1:1: \cdots: \beta^{q-2}]$). Let $W$ be the subspace of $C_{\mathbb{T}_{s-1}}(d)$ generated by $\Lambda_F$ and $\Lambda_G$. Thus $\dim W=2$ and, by using Theorem \ref{theoremzeroes}, we obtain that
\begin{align*}
|{\mbox{supp}} \, (W)| & = |{\mbox{supp}} \, \{\Lambda_F, \Lambda_G\}| 
=|\mathbb{T}_{s-1}|-|Z_{\mathbb{T}_{s-1}}(F) \cap Z_{\mathbb{T}_{s-1}}(G)| \\
& =(q-1)^{s-1}-[(q-1)^{s-1}-q+l]=q-l.
\end{align*}

Therefore
$
d_2(C_{\mathbb{T}_{s-1}}(d)) \leq q-l
$. But for these values of $d$, $d_1(C_{\mathbb{T}_{s-1}}(d))=q-1-l$ (see Equation (\ref{eqmd})). As $d_2(C_{\mathbb{T}_{s-1}}(d)) > d_1(C_{\mathbb{T}_{s-q}}(d))$ we conclude that $d_2(C_{\mathbb{T}_{s-1}}(d)) \geq q-l$. Then $d_2(C_{\mathbb{T}_{s-1}}(d))=q-l$, and the claim follows.

{\bf{Case III:}} Let $1 \leq d \leq (q-2)(s-2)$. In this case $k \leq s-3$. If we take $F$ and $G$ as the polynomials defined in the Cases I or II of the proof of Theorem \ref{theoremzeroes} (depending on the value of $d$) and, similarly to the Case II above, $U$ is the subspace of $C_{\mathbb{T}_{s-1}}(d)$ generated by $\Lambda_F$ and $\Lambda_G$, then
\begin{align*}
|{\mbox{supp}} \, (U)| & = |{\mbox{supp}} \, \{\Lambda_F, \Lambda_G\}| =|\mathbb{T}_{s-1}|-|Z_{\mathbb{T}_{s-1}}(F) \cap Z_{\mathbb{T}_{s-1}}(G)|  \\
& =(q-1)^{s-1} -[(q-1)^{s-(k+3)}[(q-1)^{k+2}-(q-1)(q-l)+1]] \\
& =(q-1)^{s-1}-[(q-1)^{s-1}-(q-1)^{s-(k+3)}[(q-1) (q-l)-1]] \\
& = (q-1)^{s-(k+3)}[(q-1)(q-l)-1].
\end{align*}

Therefore
\begin{equation} \label{upperboundGHW}
d_2(C_{\mathbb{T}_{s-1}}(d)) \leq |{\mbox{supp}} \, (U)| 
=(q-1)^{s-(k+3)}[(q-1)(q-l)-1]. 
\end{equation}

On the other hand, let $\mathcal{D}$ be a subspace of $C_{\mathbb{T}_{s-1}}(d)$ with $\dim_K \mathcal{D}=2$. If $\{\Lambda_f,\Lambda_g\}$ is a $K$--basis of $\mathcal{D}$ then, by using Theorem \ref{otherbound}, we obtain that
$$
|{\mbox{supp}} \, (\mathcal{D})| \geq (q-1)^{s-(k+3)}[(q-1)(q-l)-1]. 
$$

Therefore
\begin{equation} \label{lowerboundGHW}
d_2(C_{\mathbb{T}_{s-1}}(d)) \geq (q-1)^{s-(k+3)}[(q-1)(q-l)-1].
\end{equation}

Equations (\ref{upperboundGHW}), and (\ref{lowerboundGHW}) prove case III. The claim follows from cases I, II, and III.
\end{proof}

\begin{remark} \label{finalremark}
It is easy to see that the formulae for the second generalized Hamming weight of the codes $C_{\mathbb{T}_{s-1}}(d)$ given by Equations (\ref{eqshw3}) and (\ref{eqshw}) can be reduced to
$$
d_2(C_{\mathbb{T}_{s-1}}(d))=d_1(C_{\mathbb{T}_{s-1}}(d))+\left\lceil (q-1)^{s-(k+3)}(q-2)\right\rceil,
$$
where $s \geq 2$, $d \geq 1$, $k$ and $l$ are the unique integers such that $d=k(q-2)+l$, $k \geq 0$, $1 \leq l \leq q-2$, and $d_1(C_{\mathbb{T}_{s-1}}(d))$ is given by Equation (\ref{eqmd2}).
\end{remark}

\begin{example} \label{example}
Let $K=\mathbb{F}_5$ be a finite field with 5 elements. For the codes $C_{\mathbb{T}_2}(d)$ with $1\leq d \leq 6$ (because the $a$--invariant is 5, see \cite[Lemma 1, (II)]{sarabia0}) we obtain the complete weight hierarchy by using Proposition \ref{proposition1}, Theorem \ref{theorem1}, Theorem \ref{theorem3}, and \cite[Theorem 3]{wei} (see Tables \ref{GHWT2} and \ref{GHWT3}). 

\begin{table}[!ht]
	\centering
		\begin{tabular}{||c||c||c||c||c||c||c||c||c||} \hline
			$d$ & $d_1$ & $d_2$ & $d_3$ & $d_4$ & $d_5$ & $d_6$ & $d_7$ & $d_8$ \\
			\hline 
			1 & 12 & 15 & 16 & -- & -- & -- & -- & -- \\ 
			\hline 
			2 & 8 & 11 & 12 & 14 & 15 & 16 & -- & -- \\
			\hline
			3 & 4 & 7 & 8 & 10 & 11 & 12 & 13 & 14 \\
			\hline
			4 & 3  & 4 & 6 & 7 & 8 & 9 & 10 & 11 \\
			\hline
			5 & 2 & 3 & 4 & 5 & 6 & 7 & 8 & 9 \\
			\hline
			6 & 1 & 2 & 3 & 4 & 5 & 6 & 7 & 8 \\ \hline
		\end{tabular} 
		\vspace{0.2cm}
	\caption{The first eight generalized Hamming weights for $C_{\mathbb{T}_2}(d)$, $q=5$, $1 \leq d \leq 6$.}
	\label{GHWT2}
\end{table}

\begin{table}[!ht]
	\centering
		\begin{tabular}{||c||c||c||c||c||c||c||c||c||} \hline
			$d$ & $d_9$ & $d_{10}$ & $d_{11}$ & $d_{12}$ & $d_{13}$ & $d_{14}$ & $d_{15}$ & $d_{16}$ \\
			\hline 
			1 & -- & -- & -- & -- & -- & -- & -- & -- \\ 
			\hline 
			2 & -- & -- & -- & -- & -- & -- & -- & -- \\
			\hline
			3 & 15 & 16 & -- & -- & -- & -- & -- & -- \\
			\hline
			4 & 12  & 13 & 14 & 15 & 16 & -- & -- & -- \\
			\hline
			5 & 10 & 11 & 12 & 13 & 14 & 15 & 16 & -- \\
			\hline
			6 & 9 & 10 & 11 & 12 & 13 & 14 & 15 & 16 \\ \hline
		\end{tabular}
		\vspace{0.2cm}
	\caption{The remaining generalized Hamming weights for $C_{\mathbb{T}_2}(d)$, $q=5$, $1 \leq d \leq 6$.}
	\label{GHWT3}
\end{table}

For example, if we take $d=2$, $a_{\mathbb{T}_2}-d=3$, $d_1(C_{\mathbb{T}_2}(3))=4$ (see \cite[Theorem 3.5]{ci-codes}), and $\beta=H_{\mathbb{T}_2}(2)=6$ (see \cite[Lemma 1, (III)]{sarabia0}). Then, by Theorem \ref{theorem1},
$$
d_{6-i} (C_{\mathbb{T}_2}(2))=16-i \,\,{\mbox{for all}} \,\, i=0,\ldots,2.
$$

Moreover, if we take $d=3$, $d_1(C_{\mathbb{T}_2}(2))=8$ then
$$
d_{10-i} (C_{\mathbb{T}_2}(3))=16-i \,\,{\mbox{for all}} \,\, i=0,\ldots,6.
$$

Notice that $C_{\mathbb{T}_2}(3)$ is equivalent (we use the definition given in \cite[Remark 1]{sarabia7}) to the dual code of $C_{\mathbb{T}_2}(2)$. Therefore if we use Theorem \ref{theorem3}, and the Duality Theorem (see \cite[Theorem 3]{wei}) we obtain that
$$
d_2(C_{\mathbb{T}_2}(2))=11 \,\,\, {\mbox{and}} \,\,\, d_3(C_{\mathbb{T}_2}(2))=12,
$$
and thus we get the six Hamming weights of $C_{\mathbb{T}_2}(2)$. It is important to comment that we use Macaulay2 \cite{macaulay} to check some computations.
\end{example}

\begin{remark}\label{dec-19-17}
Let $X$ be the toric set parameterized by the edges of the complete
bipartite graph $\mathcal{K}_{m,n}$, that is, if $V_1=\{Z_1,\ldots,Z_m\}$,
$V_2=\{Z_{m+1},\ldots,Z_{m+n}\}$ is the bipartition of $\mathcal{K}_{m,n}$, then
\begin{align*}
X= \{[t_1t_{m+1} : \cdots : &\,  t_1t_{m+n} : t_2t_{m+1} : \cdots: t_2t_{m+n} : \cdots : \\
& t_mt_{m+1} : \cdots : t_mt_{m+n}] \in \mathbb{P}^{mn-1} : t_i \in K^*
\}.
\end{align*}

Equations (\ref{eqmd2}) and (\ref{eqshw}) allow us to compute the second generalized Hamming weight of the code $C_X(d)$ because this code is $C_{\mathbb{T}_{m-1}}(d) \otimes C_{\mathbb{T}_{n-1}}(d)$ (see \cite{sarabia1}). Actually
$$
d_2(C_X(d))=\min \{d_1(C_{\mathbb{T}_{m-1}}(d)) \cdot d_2(C_{\mathbb{T}_{n-1}}(d)), 
d_1(C_{\mathbb{T}_{n-1}}(d)) \cdot d_2(C_{\mathbb{T}_{m-1}}(d))\}.
$$
\end{remark}

\smallskip

\noindent {\bf Acknowledgments.} We thank the referees for a careful
reading of the paper and for the improvements that they suggested.


\end{document}